\documentclass[12pt]{article}

\usepackage[colorlinks=true, citecolor=blue]{hyperref}
\usepackage{enumerate}
\usepackage{mathtools, amssymb, latexsym, amsmath, amsfonts, amsthm, xcolor}
\usepackage{array, fancyhdr, bm, listings}
\usepackage{algorithm,comment}
\usepackage{algorithmic}
\usepackage{subcaption}
\usepackage{thmtools, thm-restate}
\declaretheorem{theorem}

\theoremstyle{plain}
\newtheorem{lemma}[theorem]{Lemma}

\theoremstyle{definition}
\newtheorem{definition}[theorem]{Definition}
\newtheorem{remark}[theorem]{Remark}



\providecommand{\keywords}[1]{\textit{Keywords:} #1}
\title{An adjacency labeling scheme based on a tree-decomposition}

\date{}

\author{Avah Banerjee}
\newcommand{\FormatAuthor}[3]{
\begin{tabular}{c}
#1 \\ {\small\texttt{#2}} \\ {\small #3}
\end{tabular}
}
\author{
\begin{tabular}[h!]{lcr}
   \FormatAuthor{Avah Banerjee}{banerjeeav@mst.edu}{Missouri S\&T}
\end{tabular}
}

\begin{document}

\maketitle

\begin{abstract}
In this paper we look at the problem of adjacency labeling of graphs. Given a family of undirected graphs  the problem is to determine an encoding-decoding scheme for each member of the family such that we can decode the adjacency information of any pair of vertices only from their encoded labels. Further, we want the length of each label to be short (logarithmic in $n$, the number of vertices) and the encoding-decoding scheme to be computationally efficient. We proposed a simple tree-decomposition based encoding scheme and used it give an adjacency labeling of size $O(k \log k \log n)$-bits. Here $k$ is the clique-width of the graph family. We also extend the result to a certain family of $k$-probe graphs.

\end{abstract}

\section{Introduction}\label{sec: intro}
Adjacency labeling is a method to store adjacency information implicitly within vertex labels such that we can determine the adjacency between two vertices just from their labels. 
To be useful in practice we want these labels to be compact and easy to encode-decode.
This is a powerful technique for lossless compression of graphs. 
Since the decoding is fully local, it makes these schemes particularly useful for storing graphs on  distributed systems.
It is an active area of research to determine adjacency labeling schemes for various graph families of practical importance.

\subsection{Preliminaries}
Let $G = (V, E)$ be an undirected graph with vertex set $V$ ($|V| = n$) and edge set $E$ ($|E| = m$). We assume $G$ has no self-loops or parallel edges. Let ${\cal F}_n$ be a family of graphs on the vertex set $V$ of size $n$. For any $u, v \in V$, we define $\mathsf{adj}(u,v) = 1$ if $\{u,v\} \in E$ and $0$ otherwise.

\begin{definition}[modified from \cite{alstrup2015adjacency}]\label{def: labeling}
An $L$-bit adjacency labeling scheme of a graph family ${\cal F}_n$ is a pair of functions $\mathsf{enc}: {\cal F}_n \to (V \to  \{0,1\}^L)$ and $\mathsf{dec}: \{0,1\}^L\times \{0,1\}^L \to \{0,1\}$ such that for all $G = (V,E) \in {\cal F}_n$ and for all $u,v \in V$, $$\mathsf{adj}(u,v) = \mathsf{dec}(\mathsf{enc}(G)(u),\mathsf{enc}(G)(v)).$$
We say there is an $L$-bit adjacency labeling for ${\cal F}_n$.
\end{definition}
We write $\mathsf{enc}(G)(u) = \mathsf{enc}(u)$ when the graph $G$ is clear from the context. According to above definition a labeling scheme is    local; as it determines the adjacency only based on the vertex labels.
\noindent For a labeling scheme $(\mathsf{enc},\mathsf{dec})$ to be useful in practice we want  both functions, $\mathsf{enc}$ and $\mathsf{dec}$, to be efficiently computable.
\noindent Here we use the qualifier ``adjacency" labeling to distinguish it from other types labeling schemes (see below). 
However, in their seminal paper, 
authors in \cite{kannan1992implicat}  referred to such a scheme simply as an $L$-labeling of $G$.  
\noindent In general the $(\mathsf{enc, dec})_{\mathsf P}$ pair may be used as an efficient  storage-retrieval scheme for ${\cal F}_n$ with respect to some predicate $\mathsf P$.
For example $\mathsf P$ could be the predicate that a triple of three vertices forms a triangle in $G$. Another example is the distance labeling problem \cite{alstrup2015sublinear} where given a pair of vertex labels the decoder outputs the shortest path distance between them.

In this paper we are only concerned with adjacency labeling.  
There is a simple yet beautiful connection between adjacency labeling and \emph{induced universal} graphs of a hereditary graph family.
\begin{definition}
A graph property $\cal P$ is said to be \emph{hereditary} if it is closed under taking induced subgraphs. 
\end{definition}
\begin{definition}\cite{alon2017asymptotically,alstrup2017optimal,abrahamsen2016near}
A graph $G_{\cal P}$ of size $f(n)$ (for some time-constructible\footnote{$f(n)$ can be computed in time $O(f(n))$.} function $f: \mathbb{N} \to \mathbb{N}$) is called universal for $\cal P$ if every graph $G \in {\cal P}$ with at most $n$ vertices is an induced subgraph of $G_{\cal P}$. 
\end{definition}
 
\noindent An adjacency labeling for a hereditary family is said to be \emph{efficient} if $k = O(\log n)$. It is an easy exercise to note that  having an efficient adjacency labeling  for a hereditary family implies that there is a induced universal graph $G_{\cal P}$ with $O(n^{O(1)})$ vertices. 
In this paper we give  an adjacency labeling for graphs parameterized over its \emph{clique-width}.
Upto a constant factor, this scheme is efficient for graphs of bounded clique-width.

\begin{definition}\cite{courcelle1993handle}\label{def: clique-width}
The \emph{clique-width} (denoted by $cw(G)$) of a graph $G$ is the minimum number of labels (of vertices) to construct $G$ using the following four operations:
\begin{enumerate}
    \item Create a vertex in $v$ with label $i$ (denoted by $(v,i)$)
    \item Disjoint union $G_1\oplus G_2$    \footnote{ The vertex set of $V(G_1\oplus G_2)$ of $G_1\oplus G_2$ is $V(G_1) \cup V(G_2)$ and the edge set $E(G_1\oplus G_2) = E(G_1)\cup E(G_2)$} of two labeled graphs $G_1$ and $G_2$
    \item  Join operation $\eta_{i,j} :$ adds edges between every pair of vertices one with label $i$ and another with label $j$ ($i \ne j$)
    \item Relabel operation $\rho_{i\to j}$ relabels vertices having label $i$ with label $j$
\end{enumerate}
\end{definition}
\noindent A construction of $G$ using the above operations is known as a $k$-expression where $cw(G) = k$.
A $k$-expression can be equivalently represented as a rooted binary tree\footnote{It is a tree and not a DAG as the same graph does not take part in two separate union operation.} $T$ (called a union tree \cite{kamali2018compact}) as follows. 
Leaves of $T$ corresponds to the labeled (with their initial labels) vertices $(v,i)$'s of $G$.
Each internal node correspond to a union operation.
Lastly, each internal node is decorated with a (possibly empty) sequence of join and relabel operations. We use the notation $d_z$ to denote the  decorator for the node $z$.

 \begin{figure}[h]
	\includegraphics[width=12cm]{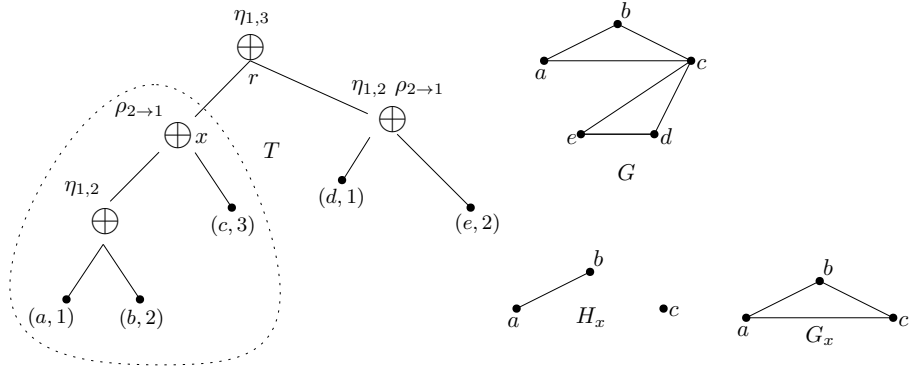}
	\centering
	\caption{$T$ is a union tree of the graph $G$. However, $T$ is not a proper union tree. The induced subgraphs $G_x = G[\{a,b,c\}]$ has edges $ac$ and $bc$ but $H_x$, the graph corresponding to the subtree $T_x$ rooted at $x$, has no such edges. }
\label{fig: proper}
\end{figure}

\noindent We say $k$ is the \emph{width} of $T$. 
 For some internal vertex $x$ of $T$ let $T_x$ be the subtree rooted at $x$. Let $G_x$ be the induced subgraph of $G$ determined by the leaves of $T_x$. 
Then $T_x$  (including any join or relabel operations in $d_x$) is a union tree for some spanning subgraph\footnote{$H$ is a spanning subgraph of $G$ if $V(H)=V(G)$ and $E(H) \subseteq E(G)$.} $H_x$ of $G_x$. Borrowing the terminology from \cite{kamali2018compact} we say $T$ is a \emph{proper} union tree of $G$ if for every internal vertex $x \in T$, $H_x = G_x$ (see example in Fig. \ref{fig: proper}).
It is an easy exercise (see lemma 1 in \cite{kamali2018compact}) to show that we can  transform any union tree in linear time to a proper one of the same width representing the same graph.
Henceforth we shall assume without loss of generality that we are working with proper union trees.

In this paper, we also look at a generalization of $k$-expressions and study adjacency labeling of the corresponding graph family. Recently, a new width parameter was proposed \cite{hung2010some,chandler2008partitioned}. 
 
 \begin{figure}[h]
	\includegraphics[width=12cm]{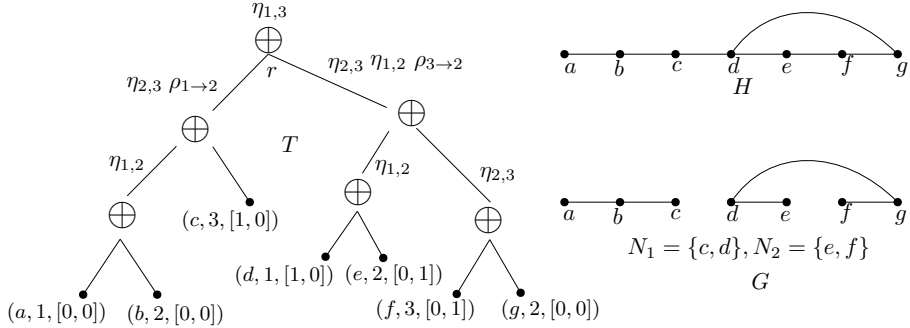}
	\centering
	\caption{Top right graph $H$ corresponds to the $k$-expression $t = \eta_{1,3}(t_1 \oplus t_2)$ where
	$t_1 = \rho_{1\to 2}\eta_{2,3}(\eta_{1,2}((a,1)\oplus(b,2))\oplus (c,3))$ and $t_2 = \rho_{3\to 2}\eta_{1,2}\eta_{2,3}((\eta_{1,2}((d,1)\oplus(e,2)))\oplus(\eta_{2,3}((f,3)\oplus(g,2))))$. $G$ can be embedded into $H$ using two independent sets $N_1, N_2$ as illustrated by the union tree $T$ of $G$.}
\label{fig: probe ex}
\end{figure}
\begin{definition}(from \cite{chang2011block})\label{def: f width}
Let ${\cal F}$ be a family of graphs. The $\cal F$-width of a graph $G$ is the minimum number $k$ of independent sets $N_1,\ldots,N_{k}$ in $G = (V,E)$ such that there exists $H = (V, E') \in {\cal F}$ where the following holds: 1)  $G$ is a spanning subgraph of $H$ and 2) for every edge $(u,v) \in E'\setminus E$  there exists an $i \in[k]$\footnote{Here $[n] = \{1,\ldots,n\}$.} with $u,v \in N_i$.
\end{definition}

\noindent A graph which has an $\cal F$-width of $k$ is known as a $k$-\emph{probe} $\cal F$-graph\footnote{Some authors call them probe-$k$ $\cal F$-graph\cite{chang2011block}}.
In this paper we consider the adjacency labeling of $w_k$-\emph{probe} ${\cal C}_k$- graphs. Here ${\cal C}_k$ is the family of graphs with clique-width $\le k$. 
We can represent a tree decomposition of 
$w_k$-\emph{probe} ${\cal C}_k$- graphs via a minor modification to the proper union tree  structure (Fig. \ref{fig: probe ex}).
The labels of each leaf now has an additional $w_k$-length binary vector. Specifically, each leaf corresponds to a tuple $(u, i, M_u)$ where $M_u[j] = 1 \iff u \in N_j$ and $i$ is $u$'s initial label in the $k$-expression (as before).
Adjacency is determined as follows (using definition \ref{def: clique-width} and \ref{def: f width}). Let $z = lca(u, v)$. 
Then $u, v$ are adjacent if and only if : 1) there is a join operation in $d_z$ between the current labels of $u$ and $v$ and 2)
 $M_u$ and $M_v$ do not have a common 1. 
 It should be noted that a $w_k$-\emph{probe} ${\cal C}_k$- graph has a clique-width $\le k2^{w_k}$. 

The motivation for studying  adjacency labelling of $k$-probe $\cal F$-graphs are threefold. Firstly, they are a generalization of probe-graphs \cite{chang2011block},  which can model some natural problems. For example a type of DNA mapping problem can be formulated as a  recognition problem for probe-graphs of intervals \cite{chandler2009probe}.
Secondly, this family of graphs do not have a bounded genus. This may make finding an compact adjacency labeling  challenging; especially if $\cal F$ does not posses tree decomposition (like a union tree). Existing approaches such as those developed in \cite{dujmovic2020adjacency} does not extended to graph families whose genus is not bounded. We leave this as an open problem.
Finally, depending on $\cal F$, many computationally hard problems exhibit efficient algorithms when the ${\cal F}$-width is bounded. For example, the recognition problem for graphs of bounded ${\cal B}$-width is fixed parameter tractable \cite{chang2011block}. Here $\cal B$ is the class of block graphs \cite{brandstadt1999graph}. A graph is a block graph if it is chordal and has no induced subgraph isomorphic to the diamond graph ($K_4-e$).

\subsection{Summary of Our Results}
Our main result is as follows.
\begin{restatable}{thm}{gentree}
\label{thm: clique width tree}
Suppose ${\cal C}_{k,n}$ is a family of graphs with $n$-vertices having a clique-width at most $k$. Then  ${\cal C}_{k,n}$ has an adjacency labeling scheme of size $O(k \log k \log n)$. If $k$ is bounded the above result is optimal upto a constant factor.
\end{restatable}

\noindent 
\noindent Briefly, we apply a recursive transformation on the union tree to obtain a tree of $O(\log n)$ depth. This transformation preserves the lowest common ancestor relations between the leaves and allows us to encode the adjacency information contained within the internal nodes and the leaves with $O(k\log k)$ bits. We also get the following generalization as a corollary.
\begin{restatable}{cor}{probe}
There is an $O(k\log k\log n + w_k)$-bits labeling scheme for $w_k$-\emph{probe} ${\cal C}_k$- graphs of size $n$.
\end{restatable}
\label{cor: gen tree}

\begin{proof}
This immediately follows from Theorem \ref{thm: clique width tree} and the fact that we need an additional $w_k$-bits to encode the vectors $M_u$.
\end{proof}



\subsection{Previous and Related Work}
Adjacency labeling schemes studied in this paper closely follows the paradigm introduced in \cite{kannan1992implicat,muller1989local}. However, the study of adjacency labeling schemes goes back more than half a century \cite{breuer1966coding,breuer1967unexpected}.
Since then many results have been discovered for a wide variety of graph classes. 
A comprehensive overview  and some interesting open problems  can be found in \cite{spinradefficient,scheinerman2016efficient}  and the references therein.
So we restrict our discussion to results which are closely related to ours.
A folklore result\footnote{We thank an anonymous reviewer for pointing this out.} is that cographs have $O(\log n)$-bit adjacency labeling. This follows from the fact that a cograph is a permutation graphs and for which an adjacency labeling follows trivially (for each vertex store $(i, \pi(i))$)\cite{spinradefficient}. In \cite{gavoille2007shorter} authors gave a $(\log n + O(k \log\log \frac{n}{k}))$-bits adjacency labeling scheme for graphs of tree-width $k$. 

Only a handful of results are known with respect to the clique-width parameter.
There is a parallel line of research based on \emph{ordered binary decision diagrams} (OBDD). OBDD's are a generalization of union trees in the setting of boolean functions. In \cite{meer2009obdd} authors gave a $O(n\frac{k^2}{\log k})$-sized, $O(\log n)$-depth OBDD with an encoding size of $O(\log k \log n)$-bits. This scheme is based on a bottom up tree decomposition approach originally introduced in \cite{miller1985parallel}. In contrast our decomposition scheme is top-down.
An improvement was proposed based on a tree-decomposition approach similar to ours\cite{kamali2018compact}. Here the author gave a $O(kn)$-sized data structure that supported $O(1)$-time adjacency quires.
A more recent result on OBDD type storage scheme for small clique-width graph can be found in \cite{chakraborty2021succinct}. 
However, these representations are not local and the adjacency queries are performed with the help of a global data structure (the OBDD or something similar).
The result closest to  ours can be found in \cite{courcelle2003query,spinradefficient}. The first paper uses the language of \emph{monodic second-order} logic. There, authors gave an adjacency labeling scheme, which in the language of this paper, translates to a label of size $O(f(k)\log n)$ bits. In their paper authors did not give an explicit expression for $f(k)$. 
In \cite{spinradefficient} (chapter 11) the author hinted at an $O(\log n)$ adjacency labeling for graphs with bounded clique-width. The proposal uses a recursive decomposition by successively finding balanced $k$-modules for any graph with clique-width $k$. Although explicit bounds were not provided with respect to $k$, we expect that working out the details can give a bound similar to ours.

\section{A caterpillar-type balanced decomposition}
In this section we give a simple balanced decomposition (discussed shortly) of a rooted tree (not necessarily binary).  Here we work with a generic rooted tree $T$ having $n$ leaves (hence $\le n-1$ internal nodes).
Later in section \ref{sec: main} we will use this result to prove theorem \ref{thm: clique width tree}.

 It is clear that every tree can be constructed starting from $K_1$ by repeatedly adding pendent edges. However, it may require $O(n)$ iterations to construct a tree with $n$ vertices. We show that each tree on $n$ leaves can be constructed within $O(\log n)$ iterations using a slightly more relaxed operation. We assume each non-root vertex $v$ has either  no children or at least two children (that is, $v$ does  not have exactly one child). Let $\mathcal L_n$ denote the set of all trees with at most $n$ leaves. See Fig. \ref{fig: trees l3} for all trees in $\mathcal L_3$, where the root is colored red. 
 \begin{figure}[h]
	\includegraphics[width=8cm]{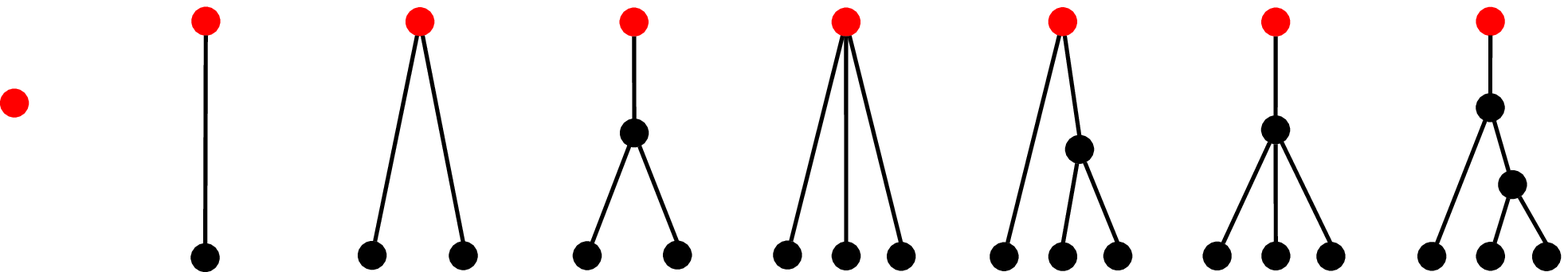}
	\centering
	\caption{Trees in ${\cal L}_3$. }
\label{fig: trees l3}
\end{figure}
 \begin{definition}
 A {\it caterpillar} (see Fig. \ref{fig:caterpillar}) is a tree for which there exists a root-leaf path $P$ such that all vertices outside $P$ are leaves. 
 \end{definition}
\begin{figure}[h]
	\includegraphics[width=4cm]{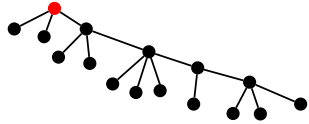}
	\centering
	\caption{A rooted caterpillar tree.}
\label{fig:caterpillar}
\end{figure}
\noindent Let $T_0,T_1,...,T_k$ be disjoint trees. The operation of {\it adding} $T_1,...,T_k$ to $T_0$ creates a tree obtained by identifying roots of $T_1,...,T_k$ with $k$ distinct leaves of $T_0$, respectively (see Fig. \ref{fig: adding}).
\begin{figure}[h]
	\includegraphics[width=9cm]{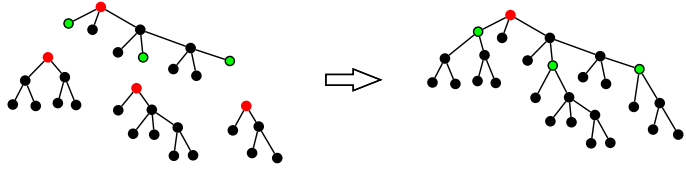}
	\centering
	\caption{Adding the three bottom trees with red roots by gluing them with the green ones of the top tree.}
\label{fig: adding}
\end{figure}
\noindent
Let $\mathcal C_0$ denote the class of all caterpillars. For each positive integer $p$, let $\mathcal C_p$ consist of trees obtained by adding trees from $\mathcal C_{p-1}$ to trees from $\mathcal C_0$. Note that $\mathcal C_{p-1}\subseteq \mathcal C_p$ since $K_1\in \mathcal C_0$.
A path $P$ of a tree $T$ is called an $r$-{\it path} if $r$ is an end of $P$. 
Let $P$ be an $r$-path of a tree $T$.
Suppose $T\ne P$. Let $T\setminus P$ denote the graph obtained by deleting $V(P)$ from $T$. Then each component of $T\setminus P$ must be one of the following two types: those that have no edges, which we call {\it trivial}, and those that have at least one edge, which we call {\it nontrivial}.
Let $T_0$ consist of all edges that are incident with at least one vertex of $P$. Then $T_0$ is a caterpillar (with root $r$). 
Suppose $T\ne T_0$. Then at least one component of $T\setminus P$ is nontrivial. 
Let $T_1,...,T_k$ be all such components. Then,
\begin{enumerate}[(i)]
    \item the root of $T_i$ is the vertex of $T_i$ that is closest to $r$ (in $T$).
    \item $E(T_0),E(T_1),\ldots,E(T_k)$ form a partition of $E(T)$
    \item $T$ can be obtained by adding $T_1,\ldots,T_k$ to $T_0$.
\end{enumerate}
The following theorem gives a structural relationship between ${\cal L}_n$ and ${\cal C}_p$.

\begin{restatable}{thm}{cater}
\label{thm: cater}
For every integer $n\ge 1$, we have $\mathcal L_n\subseteq \mathcal C_p$, where $p=\lfloor \log_2 n\rfloor$
\end{restatable}
\noindent To prove the theorem we use the following lemma.
\begin{lemma}\label{lmm: cater}
 Let $T$ be a tree with $n\ge1$ leaves. Then $T$ has an $r$-path $P$ such that each component of $G\setminus P$ has at most $n/2$ leaves.
\end{lemma}

\begin{proof}
If $n=1$ then the path with only one vertex $r$ satisfies the requirement. So we assume $n\ge2$. Under this assumption, for any $r$-path $P$, $T\setminus P$ must have at least one component. This allows us to define for any $r$-path $P$:

\begin{align*}
    h(P) = \max \{\mbox{$t\mid\ T\setminus P$ has component with $t$ leaves}\}
\end{align*}
 Let $P$ be an $r$-path that minimizes $h(P)$. We prove that $P$ satisfies the lemma. 

\begin{figure}[h]
	\includegraphics[width=4.5cm]{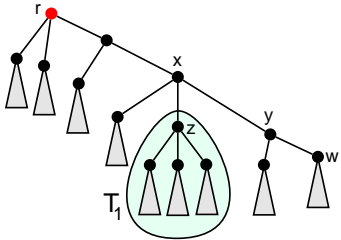}
	\centering
	\caption{The tree used in the proof of Lemma \ref{lmm: cater}}
\label{fig: lemma 1}
\end{figure}

\noindent Suppose on the contrary that $P$ does not satisfy the lemma. That is, $T\setminus P$ has a component $T_1$ with $n_1>n/2$ leaves. Let the ends of $P$ be $r$ and $w$ and let the root of $T_1$ be $z$, as illustrated in the Fig. \ref{fig: lemma 1} above. Let $Q$ be the unique path of $T$ between $r$ and $z$. We prove that $h(Q)<n_1\le h(P)$, which will be a desired contradiction.  

To estimate $h(Q)$ we observe that $T\setminus Q$ has two types of components: those that are disjoint from $T_1$ and those that are contained in $T_1$. For the ones that are disjoint from $T_1$, the number of leaves each of them may have is bounded by $n-n_1$, which is smaller than $n_1$. Next, we consider a component $T'$ of $T\setminus Q$ with $T'\subseteq T_1$. Since $T_1$ has $n_1>n/2\ge 1$ leaves, $z$ is not a leaf of $T$. By the assumption we made in the beginning of Section 2, $z$ has at least two children. It follows that $T'$ does not contain all leaves of $T_1$, which implies that $T'$ has fewer than $n_1$ leaves. Therefore, we have shown that every component of $T\setminus Q$ has fewer than $n_1$ leaves. Consequently, $h(Q)<n_1$, contradicting the choice of $P$. This contradiction proves the lemma.

\end{proof} 
\noindent {\bf Proof of Theorem \ref{thm: cater}.} 
\begin{proof}
As we observed earlier, every tree in $\mathcal L_3$ is a caterpillar, so we have $\mathcal L_n\subseteq \mathcal C_0$, for $n=1,2,3$, and thus the theorem holds for $n=1,2,3$. Suppose the theorem holds for $n-1$, where $n\ge4$. We prove that the theorem holds for $n$, and this would prove the theorem. 

Let $T$ be a tree with $n\ge 4$ leaves. We need to show $T\in\mathcal C_{\lfloor \log_2 n\rfloor}$. We may assume that $T$ is not a caterpillar because otherwise $T\in\mathcal C_0\subseteq \mathcal C_{\lfloor \log_2 n\rfloor}$. By Lemma, $T$ has an $r$-path $P$ such that each component of $T\setminus P$ has at most $n/2$ leaves. Since $T$ is not a caterpillar, $T\setminus P$ has at least one nontrivial component. Let $T_1,...,T_k$ be all such components. By our induction hypothesis, each $T_i$ belongs to $\mathcal C_{\lfloor \log_2 (n/2)\rfloor} = \mathcal C_{\lfloor \log_2 n\rfloor-1}$. It follows  that $T\in\mathcal C_{\lfloor \log_2 n\rfloor}$ since $T$ is obtained by adding $T_1,...,T_k$ to $T_0$. This completes our induction and it proves the theorem.
\end{proof}

\begin{remark}
The decomposition in the above theorem can be computed in linear time as follows. First we apply depth first search to compute for each node $u$ in $T$ the number of leaves in the subtree $T_u$.
Then we apply a  slightly modified heavy-light decomposition (see for example \cite{sleator1983data}) to obtain a  decomposition of $T$ into disjoint paths $\cal P$.
If there is a $r$-path in $\cal P$ (there can be at most one) then use it as $P$. Otherwise pick any child $u$ of $r$ and take $(r,u)$ as $P$.
We do not need to recompute the heavy-light decomposition for the recursive case and rather use the one computed for $T$. The heavy-light decomposition can be computed in linear time and hence also the caterpillar decomposition.
\end{remark}\label{rm: algo decomp}

\ifx false
\begin{theorem}\label{thm: label}
Every cographs admits a $O(\log^2 n)$-length implicit representation.
\end{theorem}

\begin{proof}
Let $T$ be the cotree corresponding to $G$. We are interested at a labeling of the leaves of $T$ from which we can infer the $\mathsf{lca}$ of pairs of leaves.
Let $T \in {\cal C}_p$ and take $P$ as our desired $r$-path of $T$.
Let ${\cal T}_k$ be the set of (possibly trivial) trees attached to the $k^{th}$ node of $P$ starting from the root $r$.
Let $v$ be any leaf in a tree from the collection  ${\cal T}_k$.
We label $v$ as follows. We use $\lceil \log n \rceil$ bits to encode the fact that $v \in {\cal T}_k$; the bits are used to store $k$.
Then we recursively compute the label of $v$ with respect to the tree ($T'$) in ${\cal T}_k$ which $v$ is a leaf of.
By construction each tree in ${\cal T}_k \in {\cal C}_{p -1}$, hence the size of the label of $v$ can be recursively computed : let $l_{p-1}(v)$ the label of $v$ with respect to $T'$. Then $l_{p}(v)$ is $l_{p-1}(v)$ appended with $\lceil \log k \rceil$-bits to specify $k$ and $\lceil \log {\cal T}_k \rceil$ bits ($\cal T$-value) to specify $T'$.
Decoding is straight forward. Given $l_p(u), l_p(v)$ we first we compare the $k$-value of $u$ and $v$. If $k_u > k_v$ or $k_v > k_u$ then $\mathsf{lca}(u,v)$ is the node on $P$ with distance $\min(k_u, k_v)$ from $r$.
Else, if their $\cal T$-value differ then $\mathsf{lca}(u,v)$ is the $k^{th}$ node from the root.
Otherwise their $\mathsf{lca}$ is determined by the labels $l_{p-1}(u), l_{p-1}(v)$.
Decoding  takes linear time in the size of the labels.
The above encoding scheme gives a simple recurrence relation for the size of the encoding:
\begin{align}\label{eqn: label rec}
  A(p) \le A(p-1) + 2\lceil \log n \rceil.  
\end{align}
which yields , $A(p) =  O(p \log n)$. This proves the claim of the theorem by taking $p \le \log n$.
\end{proof}


\fi
\section{An Adjacency Labeling Scheme}\label{sec: main}
Recall that a proper union tree $T$ is a rooted binary tree with root $r$ and has $n$ leaves. The initial label of a leaf $u$ will be denoted as $c_0(u)$. For an internal node $x \in T$ let $L_x$ be the set of leaves in the subtree $T_x$.
We begin with a lemma. 

\begin{figure}[h]
	\includegraphics[width=4cm]{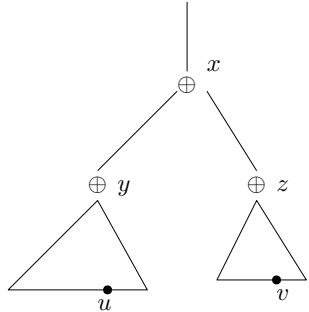}
	\centering
	\caption{We want to determine the information needed to compute the adjacency between $u$ and $v$ given we already know their lowest common ancestor $x$.}
\label{fig: local}
\end{figure}

\begin{lemma}\label{lmm: node encoding}
Suppose $G$ is a graph of clique-width $k$ and $T$ be a proper union tree of $G$. We consider a node $x$ of $T$ as shown in Fig. \ref{fig: local}.
 Let $x = \mathsf{lca}(u,v)$, where $u$ and $v$ are two leaf nodes. Given $x, c_0(u)$ and $c_0(v)$ we can determine $\mathsf{adj}(u,v)$ for all $v \in T_z$ with an additional $O(k \log k)$ bits of information stored locally at $u$.
\end{lemma}

\noindent This $O(k \log k)$-bits of information will serve to perform adjacency queries between $u$ and the set $L_z$. In theorem \ref{thm: clique width tree} we show that we can partition $V$ to $O(\log n)$ such sets for each vertex in $V$.

\begin{proof}
{\color{red} }
We consider the situation shown in Fig. \ref{fig: local}. 
Let $B_x$ be the set of unique labels assigned to the leaves of the subtree rooted at $x$ after applying $d_x$.
In order to determine $\mathsf{adj}(u,v)$ it is sufficient to know; 1) the labels of $u$ and $v$ after application of the decorators $d_y$ and $d_z$ respectively and 2) the decorator $d_x$. 
However, we do not need to know the entirety of $d_x$ but only whether $c_x(v) \in C_x(u)$, which is defined next.
Suppose $c_x(u)$ ($c_x(v)$) are the labels of $u$ (resp. $v$) before applying $d_x$. From $d_x$ we can easily determine the set of labels $C_x(u) \subseteq B_
y \cup B_z$ such that,
\begin{align*}
   \forall i \in C_x(u)\ \exists \eta_{i \to c_x(u)}\ \mbox{or}\  \eta_{c_x(u) \to i} \in d_x
\end{align*}
It is important to note that when defining the set $C_x(u)$ we consider the labels from the set $B_
y \cup B_z$ before any re-labeling due to $d_x$\footnote{Alternatively, we may assume that all relabeling operations in $d_x$ proceeds all join operations\cite{courcelle2012graph}}.
As an example, suppose $B_y = \{1, 2\}, B_z = \{3, 5\}$, $c_x(u) = 1$ and $$d_x = \rho_{3 \to 2}\eta_{1 \to 2}\rho_{2 \to 5}\eta_{5 \to 1}$$ then $C_x = \{2,3,5\}$ and not simply $\{2,5\}$.
Clearly $|C_x(u)| \le k-1$ and it takes $O(k \log k)$-bits to store $|C_x(u)|$. 
Next, we need to retrieve $c_x(v)$ for any $v \in L_z$. This can be done by storing an additional $O(k \log k)$-bits at $u$. This follows from the fact that, given an initial labeling of $L_z$, the sub-$k$-expression induced by $T_z$ (including applying the decorator $d_z$)  is just a re-labeling. This re-labeling can be stored as a list ($F_x(u)$) of size $k$ where each value is between $1$ and $k$. The $c_0(v)^{th}$ entry of this list gives $c_x(v)$.
Finally, we use $O(\log k)$ bits to store $c_x(u)$ at $u$.
\end{proof}

Going forward, we will describe an encoding of each leaf as an alternating sequence of labels of two types. One containing path information and the other containing adjacency information.
For the latter, we will use $C_x(u), F_x(u)$ and $c_x(u)$. We let $A_x(u) = (C_x(u), F_x(u), c_x(u))$.

\ifx false
Let $c_0 = c_x(u)$ and let $c_0 \to c_1 \to \ldots \to c_r$ be the sequence of re-colorings of the color of $u$ induced by $t_x$.

An edge decorator $t$ can be thought of as a  $k$-expression generated by the regular expression $\{\eta, \rho\}^*$ and appropriate labels. For example say $$t = \rho_{1 \to 4}\eta_{2,3}\rho_{4\to 2}\eta_{2,3}.$$ The size $|t|$ of the decorator $t$ is the number of join and relabel operations in its $k$-expression (in the above $|t| = 4$). We argue that $t \le {k \choose 2} + k - 1 + k/2$. First note
that a relabel operation $\rho_{i \to j}$ where the labels $j$ is present in the subtree reduces the number of labels by 1. However if $j$ is not present the number of labels does not reduce. However there can be at most $k/2$ such relabeling where the destination label is not present since otherwise due to transitivity we can combine $\rho_{i \to j}$ and $\rho_{j \to j'}$ to $\rho_{i \to j'}$. So there can be at most $k-1+ k/2 = \frac{3}{2}k-1$ relabeling operations.
Further, there are at most ${k \choose 2}$ distinct join operations (since $\eta_{i,j} = \eta_{j,i}$). Between two  successive $\eta_{i,j}$'s joining the label classes $i$ and $j$ (for example $\eta_{2,3}$ in the $k$-expression above) there must be at least one relabel of the form $\rho_{i' \to i}$ or $\rho_{j' \to j}$ otherwise we could remove one of the $\eta_{i,j}$'s without changing the result of applying the $k$-expression. 
This implies  each join operation $\eta_{i,j}$ reduces the number of pairs of label classes not yet joined by exactly 1. Hence there can be at most ${k \choose 2}$ join operations in total, giving $|t| \le {k \choose 2} + k - 1$. 
Lastly we note that it takes $2\log k + 1$ bits to  store a join or a relabel operation which consists of storing two integers in $[k]$ and the operation type. Hence we can store an edge decorator with $O(k^2\log k)$-bits (the extra $O(\log \log k)$ bits  needed to make the encoding  self-delimiting is subsumed in the big-O notation).

\fi

\gentree*

\begin{proof}

\begin{figure}[h]
	\includegraphics[width=13cm]{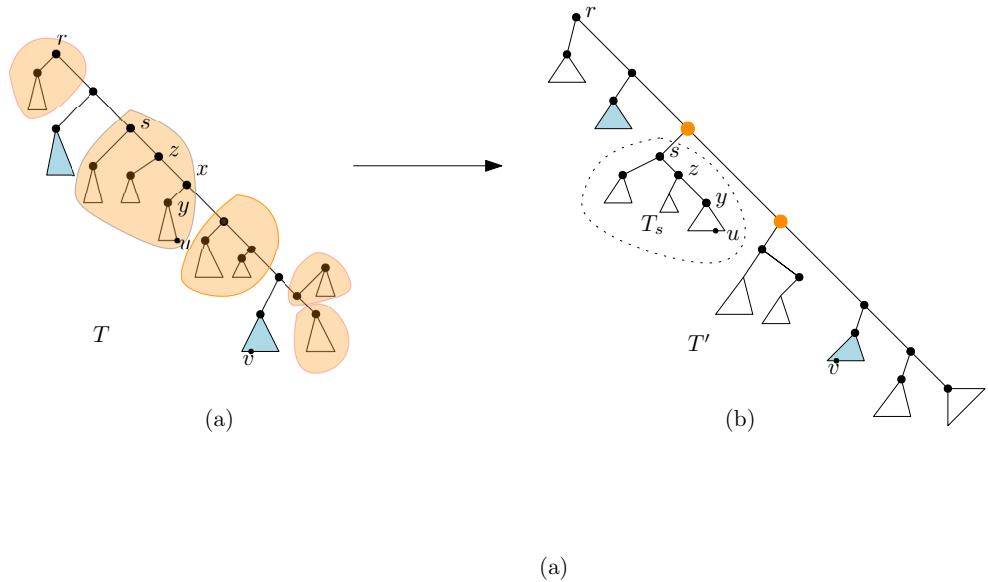}
	\centering
	\caption{The figure shows the tree obtained after contracting the smaller bushes. To make $T_s$ a proper union tree, the node $x$ is removed and $y$ is made a child of $z$.}
	\label{fig: cont}
\end{figure}

\noindent First we start from the $r$-path decomposition of $T$ as described in the previous section. Let $P$ be a $r$-path.
From Lemma \ref{lmm: cater} we know that the subtrees attached to $P$ have $\le n /2$ leaves.
Let ${\cal T}_{large}$ be a possibly empty collection of subtrees which has between $n/4$ and $n/2$ leaves.
These subtrees are identified with light blue color in Fig \ref{fig: cont}-a.
Note that $0 \le |{\cal T}_{large}| \le 4$.
  Consider the sequence(s) of smaller subtrees (we will call    them bushes) between the trees in ${\cal T}_{large}$.
These bushes are highlighted with orange regions in the figure. There may be no such bushes between two large trees. 
Now we collate the bushes between two successive large subtrees (while descending along $P$)  to create larger bushes until the total number of leaves among them is $\ge n/4$ but $\le n/2$. At this point we call it a super-bush and restart the gathering process on the remaining bushes until we get another super-bush or we reach the end of the bushes. In the latter case we create a super-bush with whatever we have gathered upto that point. That is, we group the bushes (if any) between two successive large subtrees into $\ge n/4$ sized super-bushes (but no larger than $n/2$) except may be a constant number of groups which can have  $< n/4$ leaves.
We identify each super-bush with a tree, the root of which is the node closest to $r$.
Further, we attached the tree to $P$ using the node of the super-bush which was closest to $r$.
For example, in Fig. \ref{fig: cont}-a for the  super-bush starting from the node $s$ we create a tree $T_s$ with $s$ as the root. We attach  $T_s$ to $P$ where the node $s$ was previously located.
From our construction,  the number of  such attachments will also be a constant. 
The decorators remain with the original vertices and the new ({\color{orange} orange} vertices in Fig. \ref{fig: cont}-b) vertices on $P$ does not contain any decorators. 
The resulting tree, denoted by $T'$, is not necessarily a valid union tree. However, we ensure that each subtree attached to $P$ is a proper union tree (Fig. \ref{fig: cont}-b).

First we informally describe the decoding scheme; this will give us an idea of what information to encode within the labels.
Let $u, v$ be a pair of leaves in $T$ (Fig. \ref{fig: cont}). Let $x = \mathsf{lca}(u, v)$.
From lemma \ref{lmm: node encoding} we see that labels of size $O(k \log k)$-bits are sufficient to determine $\mathsf{adj}(u,v)$ given $x, c_0(u)$ and $c_0(v)$.
It remains to be determined the number of such  labels we need to determine adjacency between $u$ and any other vertex in $G$.
Trivially, we can maintain one such label for each node on the root-leaf path (in $T$) terminating in $u$. Since a path (in the caterpillar-decomposition) can have arbitrary length we will need $\Omega(\log n)$-bits to locate a node in each level of the caterpillar decomposition. Since there are $O(\log n)$ levels, we may end up needing $O(\log^2 n)$-bits  to encode the path information in the final label.
To reduce the encoding size and get our claimed bound we make the following crucial observation.
It is not necessary to determine the $\mathsf{lca}(u,v)$ explicitly.
It suffices to know $A_x(u), c_0(u)$ and $c_0(v)$ to determine $\mathsf{adj}(u,v)$.
By taking a recursive approach, we show that we only need to remember $O(\log n)$ number of adjacency-type labels per vertex $u$. Since, each adjacency information requires $O(k \log k)$ bits labels, we get the bound claimed in the theorem. 
This recursive encoding scheme is determined based on $T'$. 
Let $l_1(u)$ (resp. $l_1(v)$) be the index of the subtree $u$ (resp. $v$) is a leaf of on the path $P$ in $T'$. 
For example, in Fig. \ref{fig: cont} we have $l_1(u) = 3$ and $l_1(v) = 5$. There are two cases:
\begin{description}
\item[(i)($l_1(u) = l_1(v)$)] Then $u,v \in L_s$ for some node $s \in P$ (see Fig. \ref{fig: cont}-b). To determine $\mathsf{adj}(u,v)$, we recurse on the subtree $T_{s}$. According to our construction $T_s$ is a proper union tree corresponding  to the induced subgraph $G[L_s]$. Then we determine an adjacency labeling scheme for $G[L_s]$ using $T_s$. This is used to determine $\mathsf{adj}(u,v)$. This recursive construction is possible, since any induced subgraph of $G$ has a clique-width $\le k$ and $T$ is a proper union tree.
\item[(ii)($l_1(u) \ne l_1(v)$)] We assume without loss of generality that $l_1(u) < l_1(v)$ (the case $l_1(u) > l_1(v)$ is symmetric). In this case we use $A_x(u), c_0(u)$ and $c_0(v)$ to determine $\mathsf{adj}(u,v)$.
This completes the informal description of the decoding. From this, an encoding scheme emerges naturally. 
\end{description}


{\it Encoder:} 
Generate $T'$ from $T$ and for each $u \in V$ we compute $(l_1(u), A_{P(u)}(u))$. Here, $P(u)$ is the lowest ancestor of $u$ on the path $P$ (in Fig. \ref{fig: cont}-b $P(u) = x$). For notational simplicity we denote $A_{P(u)}(u) = A_{1}(u)$. 
Then, perform the encoding recursively on each induced subgraphs of $G$ corresponding to the subtrees attached to $P$.
For the vertex $u$ this process gives a sequence of labels $((l_i(u),A_i(u))$'s. Appending to this sequence its initial label $c_0(u)$ gives the final encoding:  $$\mathsf{enc}(u) = (c_0(u), (l_1(u),A_1(u)),\ldots,(l_p(u),A_p(u)) ),$$ where $p = O(\log n)$ (from Theorem \ref{thm: cater}).
It is clear from the construction that $\mathsf{enc}$ uses $O(k\log k\log n)$-bits. 

{\it Decoder: }
Given two strings $\mathsf{enc}(u)$ and $\mathsf{enc}(v)$ first we check the  labels\\ $(l_1(u),A_1(u))$ and $(l_1(v),A_1(v))$.
If $l_1(u) < l_1(v)$ then we use $A_1(u), c_0(u)$ and $c_0(v)$ to determine $\mathsf{adj}(u,v)$.
The case $l_1(u) > l_1(v)$ is symmetric.
Otherwise, $l_1(u) = l_1(v)$. In this case we proceed to check the next pair of labels $(l_2(u),A_2(u))$ and $(l_2(v),A_2(v))$ and so on.
In general, let $i$ be the smallest number such that $l_i(u) \ne l_i(v)$. Then, using either $A_i(u)$ or $A_i(v)$ and $c_0(u), c_0(v)$ we can determine $\mathsf{adj}(u,v)$.
By our construction there is always such an $i \le p$ such that $l_i(u) \ne l_i(v)$.

{\it Correctness:} We induct on the depth of the recursive construction. For the base case, we take $p = 0$ and the correctness follows trivially.
Assume the encoder-decoder works correctly whenever the decomposition has depth $\le p-1$. This takes care of the case $l_1(u) = l_1(v)$. For the remaining case assume $l_1(u) < l_1(v)$. Then correctness follows from lemma \ref{lmm: node encoding}.
\end{proof}


\begin{remark}
Recall from Remark 1, given a union tree $T$ we can determine the recursive decomposition in $O(|T|)$ time. Additionally, $O(k \log k \log n)$ time is spent processing each leaf of $T$. Thus $\mathsf{enc}$ can be computed in $O(kn \log k \log n)$ time. 
Decoding can be done in linear time in the size of the labels (i.e., in $O(k \log k \log n)$ time). 
\end{remark}

\section*{Acknowledgement}
The author would like to thank Guoli Ding for many discussions and considerable advice. In particular for the proof of caterpillar decomposition. We also thank anonymous reviewers for their helpful comments.

\ifx false
\begin{figure}[h]
	\includegraphics[width=7cm]{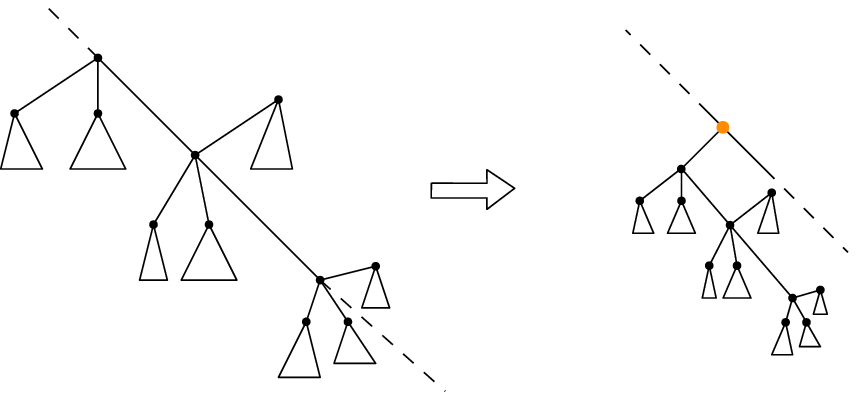}
	\centering
	\caption{Collating multiple bushes along $P$}
	\label{fig: line}
\end{figure}

\begin{figure}[h]
	\includegraphics[width=8cm]{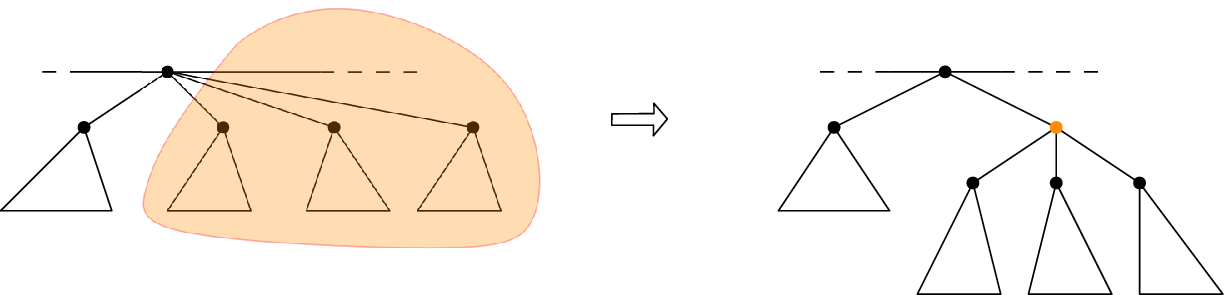}
	\centering
	\caption{Collating subtrees sharing a common representative.}
	\label{fig: multi}
\end{figure}

\section{Extension to the General Case}
Now we revisit the general case with the techniques of section 4 at hand. There are two main complications and these are shown in Fig. \ref{fig: line} and Fig. \ref{fig: multi}.
First we look at the situation of Fig. \ref{fig: line}.
Here a partial section of the $r$-path $P$ is shown.
This section has been identified as a super-bush.
In the right subfigure we see the section after transformation. This situation is almost identical to that in the previous section.
Next we look at Fig. \ref{fig: multi}.
In the left subfigure suppose the subtree to the left is already too big. Grouping it with any other subtree on its right would create a component with more than $n/2$ leaves.
In such a case we can collate the smaller subtrees (inside the orange shade) to create one or more components no more than $n/2$ leaves (with at least $\alpha n$ leaves if possible).
Applying this process throughout $T$ we can create another tree $T'$ (with respect to the same $r$-path) with the following properties.\\
(i) The maximum degree of $T'$ is bounded by a constant.\\
(ii) Total number of subtrees attached to the $r$-path $P$ is bounded by a constant.\\
(iii) Each subtree has at most $n/2$ leaves.
Hence we can apply the encoding procedure of the previous section with one additional modification.
Each sub-label now consist of three components $(l_i(u), s_i(u), p_i(u))$. The first and the third component is the same as before (index on $P$ of the root and its parity).
Here $s_i(u)$ stores the index of the subtree containing $u$ among the subtrees which are attached to $P$ via a common node (a.k.a. the representative).
By property (i) the value of $s_i(u)$ is bounded by a constant. 
Hence each sub-label needs a constant number of bits and $L(u)=((l_1(u), s_1(u), p_1(u))\ldots,(l_p(u), s_p(u), p_p(u)))$ requires $O(\log n)$ bits as before.
Note that in this case we do not have to do any special processing for the tail as it is taken care of by the label $s_i(u)$.

Now we describe the decoding. Let $L(u), L(v)$ are two labels. Suppose $l_1(u) < l_1(v)$ or $l_1(u) > l_1(v)$
then we proceed as in the previous section.
Else, suppose $l_1(u) = l_1(v)$. If $s_1(u) \ne s_1(v)$ then the adjacency is determined by either $p_1(u)$ or $p_1(v)$ (they will be the same).
Otherwise, we move on to check their second order labels $l_2(u)$ and $l_2(v)$ and so on. Here we alternate between checking the $l$-labels and the $s$-labels until we find a level where their labels differ (which by construction must exist).
This leads us to the main theorem of the paper:

\begin{theorem}
Every cograph has a $O(\log n)$-bit implicate local representation.
\end{theorem}
\fi
\ifx false

we only need to identify from the green nodes the node which is the root of the subtree $u$ is a leaf of.
These were the nodes which were originally on $P$ in $T$.
Directly storing this information in the label proves to be too costly.
Instead, we recursively decompose $T_i$, with respect to the root $r_i$.
From the previous section we know that there are $p = O(\log n)$-levels in the recursion. 
Let $L(u) = (l_1(u),\ldots,l_p(u))$.

be the tree after contracting the super-bushes. This tree is shown in the right subfigure. Let $P'$ be the contracted path corresponding to $P$.
Since $|P'| \le c+c' = O(1)$, we can label the nodes in $P'$ with a labels of size $\beta_1 = O(\log (c+c'))$, a constant. Note that, for some super-bushes we need to remember two labels. To make the labeling scheme uniform, if a super-bush only has one contracted node or if the node represents a large subtree then we keep two copies of the node. So that at the first level each super-bush or a large subtree ($T_i$) is represented by a pair $(l_i,r_i)_1$ of labels of constant size (where $r_i = l_i$ or $r_i = l_i + 1$).
Hence forth we will not differentiate between the super-bushes and the large subtrees and simply treat them as subtrees. Further each such subtree is represented by at most two nodes on the contracted path $P'$.

If two leaves of $T$ are on different subtrees, then this can be inferred from looking at their labels for the first level ($(l_i,r_i)_1$).
Suppose the leaves $u, v$ belongs to the subtrees with labels $(l_i,r_i)_1$ and $(l_j,r_j)_1$ respectively.
Then $\mathsf{lca}(u,v) = r_{\min (i,j)}$.
Now suppose $i = j$, that is both $u$ and $v$ are on the same subtree $T_i$ with respect to the decomposition at the first level.
Then we recurse on $T_i$ to find $\mathsf{lac}(u,v)$. 
In the recursive step, for each subtree $T_i$ we treat the node identified with the label $l_i$ as the root of $T_i$. Note that node with the label $r_i$ will only have one child. To make it proper we add a leaf directly to it. This does not affect the analysis for the label size, and it is only their to maintain the invariant that internals nodes have degree $\ge 3$ for a cotree.  

According to our construction each subtree is of size $\le n/2$, so we get an improved version of Eq. \ref{eqn: label rec} for the label size:

\begin{align}\label{eqn: label rec}
  A(n) \le A(n/2) + \beta_1.  
\end{align}
This immediately gives a $O(\log n)$ labeling scheme.
The label of a leaf $u$ is represented be a tuple $L(u) = \left((l_1(u),r_1(u)),\ldots,(l_p(u),r_p(u))\right)$ where $p = O(\log n)$ and size of each sub-labels is constant.
Decoding is straightforward.
Given $L(u)$ and  $L(v)$ we proceed as follows.
If $l_1(u) = l_1(v)$ (which also implies $r_1(u) = r_(v)$) we proceed recursively on the common subtree of both $u$ and $v$ given by the labels $L'(u) = \left((l_2(u),r_2(u)),\ldots,(l_p(u),r_p(u))\right)$ and  $L'(v) = \left((l_2(v),r_2(v)),\ldots,(l_p(v),r_p(v))\right)$.
If $l_1(u) \ne l_1(v)$, assume without loss of generality that $l_1(u) < l_1(v)$.
In that case we know that $\mathsf{lca}(u,v)$ is a node $z$ on the path $P$, corresponding to the subtree $u$ is in.
This node $z$ is inside the subtree

First the recursive 
: check the first sub-label if they are already different then stop. Else proceed to the next one. clearly this takes linear in the size of the labels.
From the preceding discussions we the following theorem as a consequence.

Let $f(n): \mathbb{N} \to \mathbb{N}$ such that $f(n) \in o(n)$. Let ${\cal T}_{large}$ be the collection of subtrees which has $\ge f(n)$ leaves. Note that $|{\cal T}_{large}| \le n/f(n)$. 
If ${\cal T}_{large}$ empty we simply use the procedure described in the proof of  theorem \ref{thm: label}.
Then the size of the labels, based on a recurrence recurrence relation given in equation \ref{eqn: label rec}, will be determined by the following recurrence:
\begin{align}
    A(n) \le A\left(\frac{n }{ f(n)}\right) + O(\log n)
\end{align}
Now assume ${\cal T}_{large}$is not empty.
Between two successive trees in ${\cal T}_{large}$ we identify the subtree of $T$ (shaded in orange in left diagram of Fig \ref{fig: cont}) and contract them to one  single super-node (right diagram of Fig \ref{fig: cont}). Let $T'$ be this contracted tree. 

In this section we extend the above Caterpillar decomposition in a non-trivial manner to reduce the size of the labels.
The main idea is as follows: instead of working with a path, we will use a tree which has limited number of branches. 
We  call them \emph{columnar trees}.
As in the previous section we identify the root vertex with $r$ for the tree $T$.
Now consider any rooted tree, where we allow the possibility of internal nodes having only two neighbours. 
Let the branching height of $T$ be defined as the height of the tree $T'$ which is obtained from $T$ be contracting edges incident to degree-2 internal nodes until all internal nodes are of degree $> 2$.
We use $bh(T)$ to denote the branching height of $T$.

Given our input tree $T$ let $T'$ be a rooted subtree with the same root $r$.
Let $bh(T') \le  k$. Let $F$ be the forest $T \setminus T'$ and $h(F)$ be the size (number of leaves) of the largest component in $F$.
We choose $T'$ such the $h(F)$ is minimum and $bh(T') \le k$.
Observe that Theorem 1 basically says that we can find a $T'$ with $h(F) \le |T|/2$ and $bh(T') = 1$. Since the branching height of a path is 1.
On the other extreme if allow $bh(T') \ge |T|/2$ then it is easy to see that we can find a $T'$ with $h(F) = 1$. 
\begin{lemma}
There is a $T'$ for which the following holds: 1) $bh(T') = O(\log n)$ and 2) $h(F) = O(\sqrt{n})$.
\end{lemma}

\begin{theorem}
There is a $O(\log n \log \log n)$ labeling scheme for $T$.
\end{theorem}
%
%
%

\fi

\bibliographystyle{plain}
\bibliography{ref}

\end{document}